\documentclass{ifacconf}

\usepackage{graphicx}      
\usepackage{natbib}        

\usepackage{amsmath} 
\usepackage{amssymb}  
\usepackage{float}
\usepackage{mathrsfs}
\usepackage{mathtools}

\newtheorem{theorem}{Theorem}[section]

\newcommand{\Eff}{{\rm{Eff}}}

\newcommand{\leftparen}{\left(}
\newcommand{\leftcurly}{\left\{}

\newcommand{\rightparen}{\right)}
\newcommand{\rightcurly}{\right\}}

\begin{document}
\begin{frontmatter}

\title{
Randomization to Balance Performance under Nominal and Compromised Communication\thanksref{footnoteinfo}}

\thanks[footnoteinfo]{Research was sponsored by the Air Force Office of Scientific Research under award number FA9550-23-1-0171, by the Army Research Office under grant number W911NF-25-1-0239, and by NASA under grant number 80NSSC25M7102.}

\author[UCCS]{Colton Hill}
\author[UCCS]{Brandon Collins}
\author[UCCS]{Philip N. Brown}

\address[UCCS]{Computer Science Department, University of Colorado at Colorado Springs, CO 80918, USA.
(e-mail: \{chill13, bcollin3, pbrown2\}@uccs.edu).
P. N. Brown performed the work while a Visiting Professor at Politecnico di Torino, Italy.}

\begin{abstract}                
It is well-known that under nominal communication, where agents are able to share information with one another, game-theoretic methods provide system-level performance guarantees in multiagent cyber-physical human systems such as distributed resource allocation.
However, approaches that depend on communication are frequently not robust when communication is compromised.
While randomized decision-making can prevent arbitrarily poor outcomes under compromised communication, these performance guarantees often do not improve when communication abilities are restored.
Thus, it remains unclear whether a single approach can perform well under nominal communication and remain robust under compromised communication.
In this paper, we study whether a system operator can inform agents about how much they should randomize their choice when selecting actions in order to balance performance across both settings.
We show that for any scenario (e.g., a multiagent resource allocation instance), there exists a level of randomization where nominal efficiency is near-optimal while compromised efficiency is strictly positive.
We then show that this guarantee weakens as agents must choose between increasingly many resources.
These results provide an initial step toward understanding randomization as a tool for robust performance across communication settings.
\end{abstract}


\end{frontmatter}

\section{Introduction}
Multiagent systems underpin a wide range of cyber-physical human applications, including power grids, mixed-autonomy vehicle swarms, and network routing~\citep{dorri2018multi}.
Centralized control is often infeasible, so decentralized control approaches, such as game-theoretic frameworks, offer effective methods for analyzing and designing distributed coordination in these problems~\citep{saad2012game, xu2011opportunistic, kordonis2019mechanisms, ferguson2021effectiveness}.
Fundamentally, each agent (whether machine or human) makes decisions by using some algorithm to optimize a utility function
~\citep{gopalakrishnan2011architectural, marden2012revisiting}.
This has been well-studied both in algorithm design for distributed cyber-physical human systems and as a model for human decision-making~\citep{foster1990stochastic, blume1993statistical, papadimitriou2001algorithms, alos2010logit, marden2013distributed, papadimitriou2016nash}.

A common game-theoretic paradigm models multiagent cyber-physical systems interactions as a \emph{submodular maximization game}.
Here, many utility functions render all system-optimal action profiles as Nash equilibria, and ensure that all Nash equilibria achieve a system objective value within a factor of $2$ of optimal.
That is, the \emph{price of anarchy}---the worst-case ratio of the system-level objective at Nash equilibrium, relative to the optimal system-level objective---is $1/2$~\citep{vetta2002nash}.
When these utility functions are paired with deterministic algorithms, convergence to a Nash equilibrium is guaranteed~\citep{blume1993statistical}.
However, this guarantee typically depends on a specific information assumption: each agent can monitor other agents' actions through their contributions to the system-level objective.
Under \emph{nominal} communication settings, where agents can observe the actions of others, the above results are obtained.
However, when communication is \emph{compromised} and agents can only observe their own actions, the price of anarchy worsens from~$1/2$ to~$1/n$ (where $n$ is the number of agents)~\citep{grimsman2020cost}.

Communication losses between agents are common in practice, e.g., from terrain effects, adversarial jamming, and weather.
It is well-known that when agents cannot monitor each other's actions, game-theoretic approaches are not robust~\citep{brown2018projecting}.
However, the Nash equilibria in these instances are fragile, as recent work has suggested that randomness in agents' strategies can improve system-level performance~\citep{seaton2022all,seaton2023intrinsic}.
This fact can be utilized to create randomized algorithms that drive emergent behavior to an expected system-level objective that is within~\mbox{$\sim \hspace{-0.55mm} 63\%$} of optimal.
However, these algorithms depend only on resource values and not agents' utility functions, so this performance guarantee does not improve when agents are not compromised.
Thus, two distinct approaches emerge:
\begin{enumerate}
    \item Under nominal communication, utility functions can be paired with deterministic algorithms that provide desirable system-level performance guarantees.
    
    \item Under compromised communication, randomized algorithms can prevent system-level performance from being arbitrarily poor.
\end{enumerate}

Accordingly, a system operator may have access to a utility function and learning rule pair that performs well under nominal communication and a randomization algorithm that is robust when communication is compromised.
However, if communication status changes unexpectedly and the operator is unable to change the approach being deployed in real time, the following tradeoff arises.
Select the first approach and risk losing performance guarantees when communication is compromised, or select the second approach and accept moderate performance under nominal communication.
This raises the following question:

\emph{Within a single utility-function and learning-dynamics framework, can a system operator provide a randomization guideline so performance under nominal communication is near-optimal and arbitrarily poor performance is avoided under compromised communication?}

In this paper, we study this question in the context of a \emph{submodular maximization} game, where agents evaluate resources with the \emph{marginal-contribution} utility function and optimize them via \emph{log-linear learning}~\citep{marden2012revisiting}.
We seek to understand how a system operator can select temperature (which controls how much agents randomize over suboptimal actions under log-linear learning) so that nominal performance is near-optimal while compromised performance is prevented from being arbitrarily poor.
We provide a result on the robustness of this approach, as well as a limitation.

Theorem~\ref{thm:exists_tradeoff_temperature} presents the robustness result, and shows that when agents use---or their behavior can be modeled by---marginal-contribution utility functions that are optimized via log-linear learning, temperature can be used to balance nominal efficiency with robustness to compromised communication.
That is, for any fixed game, a system operator can choose one randomization level so that nominal performance is near-optimal and compromised performance remains strictly positive.
Theorem~\ref{thm:temperature_selection_arbitrarily_bad} presents a limitation of this approach, and shows that for any fixed temperature, there exists a game with arbitrarily poor asymptotic efficiency under compromised communication.
Thus, the strength of Theorem~\ref{thm:exists_tradeoff_temperature} is agent--resource dependent, and can weaken when agents are asked to choose among too many resources.
A high-level summary is given in Fig.~\ref{fig:High_Level_Results}.

\begin{figure}[t]
    \vspace{2mm}
        \begin{center}
             \includegraphics[width=0.45\textwidth]{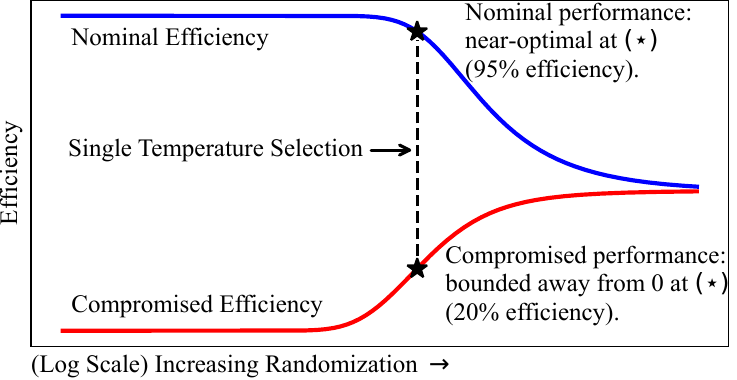}
        \end{center}
        \caption{
        The plot shows efficiency versus how much agents randomize over actions under nominal and compromised communication.
        When randomization is low, nominal efficiency is near-optimal, but compromised efficiency is arbitrarily poor.
        As randomization increases, compromised efficiency improves at the expense of nominal efficiency.
        The dashed line demonstrates Theorem~\ref{thm:exists_tradeoff_temperature} for a specific game instance: one temperature can be used under log-linear learning such that compromised efficiency is improved by $20\%$ while nominal efficiency remains within $5\%$ of optimal for the game.
        Theorem~\ref{thm:temperature_selection_arbitrarily_bad} shows this tradeoff is dependent on the number of resources, and a temperature that balances both settings cannot be found if agents are asked to choose among too many resources.
        }
        \label{fig:High_Level_Results}
    \vspace{-1mm}
\end{figure}

\section{Model}

\subsection{Multiagent Singleton Set Cover}

A multiagent resource allocation problem has agent set $N =\left\{1, \ldots, n\right\}$ and finite set of resources $\mathcal{R}$ (e.g., facilities to cover).
Each agent \mbox{$i \in N$} selects one resource from its admissible \emph{action set}, $\mathcal{A}_i \subseteq \mathcal{R}$.
The \emph{joint action space} is denoted \mbox{$\mathcal{A}:=\mathcal{A}_1 \times \dots \times \mathcal{A}_n$}, where \mbox{$a \in \mathcal{A}$} is an action profile.
We write~$a_i$ to denote the action selected by agent~$i$ in~$a$, and~$a_{-i}$ denotes the joint actions of all agents except for~$i$, i.e., \mbox{$a_{-i} = (a_1, a_2, \ldots, a_{i-1}, a_{i + 1}, \ldots, a_n)$}.
Thus, $a = (a_i, a_{-i})$, and \mbox{$R(a) = \bigcup_{i \in  N} a_i \subseteq {\mathcal{R}}$} denotes the set of resources selected in action profile~$a$.
The function \mbox{$v:\mathcal{R}\to \mathbb{R}_{\geq 0}$} maps each resource to a non-negative real number representing a resource's value.
Then, the system-level objective function, $W:\mathcal{A} \to \mathbb{R}_{\geq 0}$, is defined as
\begin{equation}
    \label{eq:Objective_Function}
    W(a):=\sum_{r \in  R(a)} v(r).
\end{equation}
To avoid trivialities, we assume $W(a)>0$ for some $a \in \mathcal{A}$, and the empty action profile $\emptyset$ satisfies $W(\emptyset) = 0$.

\subsection{Marginal-Contribution Utility Function}

To model how agents select resources with the goal of maximizing the value of the objective function in~\eqref{eq:Objective_Function}, a system operator may associate agents with a utility function $U_i : \mathcal{A} \to \mathbb{R}$ that guides their local decision-making.
A utility function provides a subjective valuation of each agent's task and reflects all information the agent has regarding the choices of other agents in the system.
We consider the \emph{marginal-contribution} utility function:
\begin{equation}
    \label{eq:marginal_contribution}
    U_i(a_i, a_{-i}) \coloneqq W(a_i, a_{-i})-W(\emptyset, a_{-i}).
\end{equation}
Intuitively, \eqref{eq:marginal_contribution} measures the change in $W$ at $a$ due to agent~$i$'s presence.
Since agents individually optimize their own utility function, a non-cooperative game is formed.

\subsubsection{Blind Agents}

When all agents are able to observe the actions of others, marginal-contribution utility functions, given by~\eqref{eq:marginal_contribution}, render all system-optimal action profiles as Nash equilibria~\citep{wolpert2001optimal}.
However, agents can suffer a fault that compromises their observation abilities.
Thus, recent work has extended the above model, and shows that Nash equilibria can be arbitrarily bad when agents are \emph{blind}, and are unable to observe the actions of others~\citep{grimsman2020cost}.
We write $B \subseteq N$ to denote the set of blind agents; each agent $i \in B$ receives no information about the actions of other agents.
To formally model communication loss for each blind agent $i \in B$, the utility function is modified to become:
\begin{equation}
    \label{eq:utility_function_comm_denied}
    \tilde{U}_i(a_i, a_{-i}) \coloneqq U_i(a_i, \emptyset) = W(a_i,\emptyset) - W(\emptyset, \emptyset),
\end{equation}
where $W(\emptyset, \emptyset) = 0$.
Intuitively,~\eqref{eq:utility_function_comm_denied} suggests that blind agents act as though other agents are not present.
Further, agent \mbox{$i \in I \coloneqq N \setminus B$} is referred to as \emph{informed}, and receives information about all other agents.

An instance of a game of the above form is specified by the tuple \mbox{$G = \leftparen N, B, \mathcal{R}, \mathcal{A}, W, \leftcurly U_i \rightcurly_{i \in N} \rightparen$}.
We refer to \mbox{$G_B \coloneqq ( N, N, \mathcal{R}, \mathcal{A}, W, \{ \tilde{U}_i \}_{i \in N} )$} as the compromised version of $G$, where all agents are blind, and refer to \mbox{$G_I \coloneqq \leftparen N, \emptyset, \mathcal{R}, \mathcal{A}, W, \leftcurly U_i \rightcurly_{i \in N} \rightparen$} as the nominal version of $G$, where all agents are informed.
Finally, we write $\mathcal{G}$ to refer to the class of multiagent singleton set cover games.

\subsection{Log-Linear Learning}

Let $a(t)$ denote an action profile for timestep \mbox{$t \in \{0, 1, 2, \ldots \}$}.
For each $t > 0$, agent $i \in N$ is chosen at random and allowed to update their current action selection.
The updating agent selects a new action according to the probability:
\begin{equation}
    \label{eq:LLL}
    \Pr[a_i(t) = a_i] =
    \frac{e^{\frac{1}{\tau} U_i^G(a_i, a_{-i}(t-1))}}{\sum_{a_i^\prime \in \mathcal{A}_i} e^{\frac{1}{\tau} U_i^G(a_i^\prime, a_{-i}(t-1))}},
\end{equation}
where $a_i \in \mathcal{A}_i$, with temperature $\tau > 0$, and $U_i^G(a_i, a_{-i})$ is the effective utility of agent~$i$ in $G$:
\begin{equation}
    \label{eq:effective_utility}
    U_i^G(a_i, a_{-i}) \coloneqq
    \begin{cases}
        U_i(a_i, a_{-i}), & i \in I,
        \\
        \tilde{U}_i(a_i, a_{-i}), & i \in B.
    \end{cases}
\end{equation}
Intuitively, $\tau$ dictates how likely agent~$i$ is to select a suboptimal action.
As $\tau \to  \infty$, all actions available to agent~$i$ are equally likely to be selected.
As $\tau \to 0$, agent~$i$ will select a best response to the action profile $a_{-i}(t-1)$ with arbitrarily high probability (or uniformly at random in the case a best response is not unique).
All remaining agents simply repeat their current action from the previous timestep, so $a_{-i}(t) = a_{-i}(t-1)$.

We refer to any procedure for selecting the temperature parameter $\tau$ as a \emph{temperature selection scheme}.
For $\mathcal{G}$, a temperature selection scheme is a map $\tau : \mathcal{G} \to \mathbb{R}_{>0}$.
For $G \in \mathcal{G}$, the selected temperature is $\tau(G)>0$, which we often write simply as $\tau$ when clear from context.

\subsection{Expected System-Level Objective Relative to Optimal}

The system operator's task is to select $\tau > 0$ such that performance is high under nominal communication and robust to communication failures.
Since the action space is finite and log-linear learning assigns positive probability to every action at each positive temperature, the induced Markov chain admits a stationary distribution, denoted~$\pi_\tau^{G}$.
The expected value of the system-level objective is
\begin{equation}
    \label{eq:exp_obj_val}
    E_{a \sim \pi_\tau^G}[W(a)] = \sum_{a \in \mathcal{A}} W(a)\,\pi_\tau^G(a).
\end{equation}
Then, for game $G$, we define the asymptotic efficiency of the system-level objective in expectation relative to an optimal action profile, $a^* \in \arg\max_{a \in \mathcal{A}} W(a)$, as
\begin{equation}
    \label{eq:Efficiency}
    \Eff_\tau(G) \coloneqq \frac{E_{a \sim \pi_\tau^G}[W(a)]}{W(a^*)}.
\end{equation}

\subsection{Star-Graph Game}\label{sec:star_game}

\begin{figure}[t]
    \vspace{2mm}
        \begin{center}
             \includegraphics[width=0.45\textwidth]{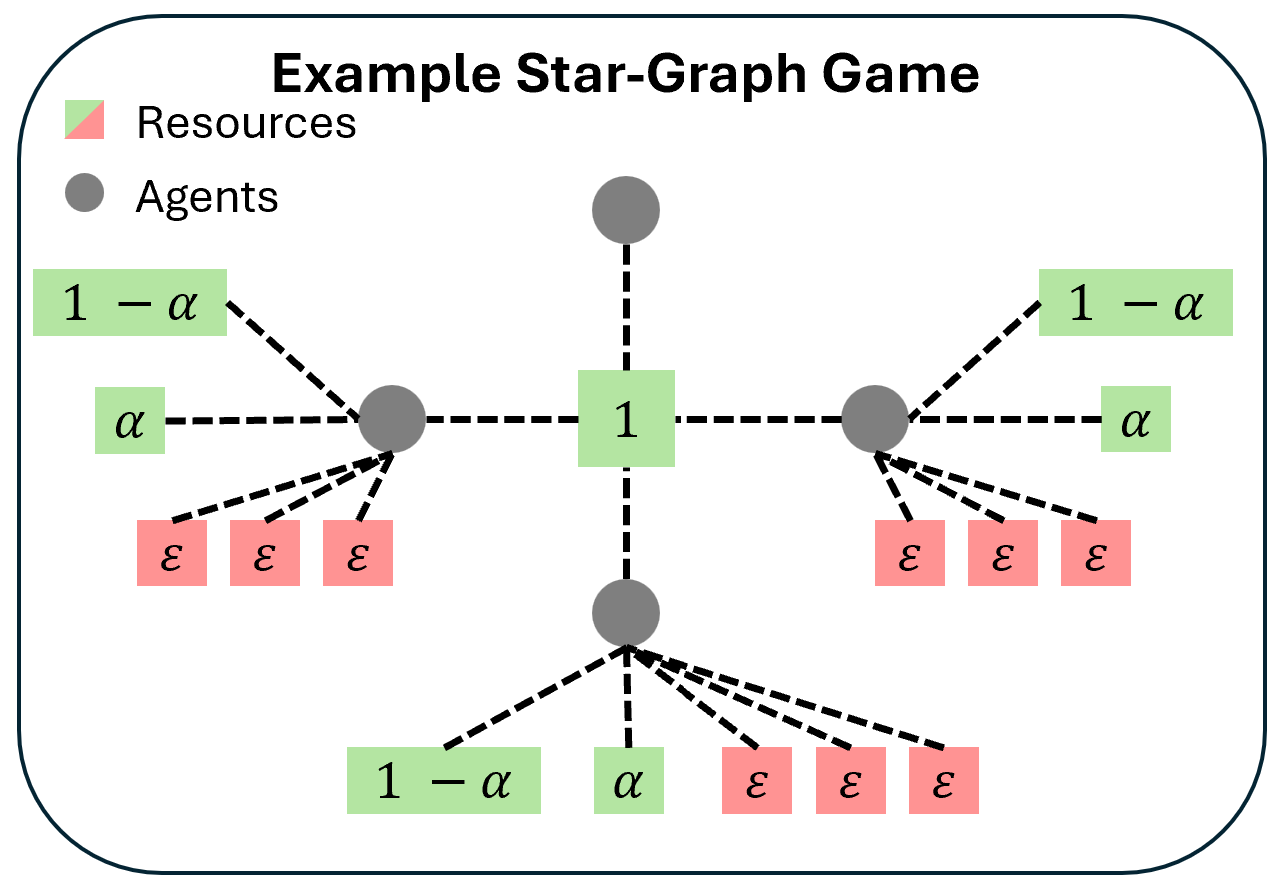}
        \end{center}
        \caption{
        $G^\star(m, n)$, where $m = 3$ and $n = 4$.
        Four agents are assigned to cover a star resource graph, where one agent's sole task is to cover a high-valued resource, and the remaining agents cover a subset of resources.
        }
        \label{fig:Example_Game}
    \vspace{-1mm}
\end{figure}

The game constructed below is used to prove Theorem~\ref{thm:temperature_selection_arbitrarily_bad} and in the numerical experiments in Section~\ref{subsec:Num_Experiments}.
It models a central high-value resource and several remote alternatives, where all agents have access to the central resource to help ensure coverage.
The set of resources consists of:
\begin{enumerate}
    \item one high-valued resource, denoted $r_1$, that is assigned a value of $1$, so $v(r_1) = 1$;
    \item $n-1$ medium-valued resources, denoted~$r_2, \ldots, r_n$, and assigned a value of $1-\alpha$, where $\alpha \in [0,0.5]$, so 
    \begin{equation}
        v(r_2) = v(r_3) = \cdots = v(r_n) = 1-\alpha;
    \end{equation}
    \item $n-1$ low-valued resources, denoted~$r_{n + 1}, \ldots, r_{2n-1}$, and assigned a value of $\alpha$, so
    \begin{equation}
        v(r_{n + 1}) = v(r_{n + 2}) = \cdots = v(r_{2n-1}) = \alpha;
    \end{equation}
    \item $m(n-1)$ negligible-value resources for $m \in \mathbb{N}$, denoted $r_{i,1}^{\varepsilon}, r_{i,2}^{\varepsilon}, \ldots, r_{i,m}^{\varepsilon}$ and assigned a value of $\varepsilon > 0$, so
    \begin{equation}
        v(r_{i,1}^{\varepsilon}) = v(r_{i,2}^{\varepsilon}) = \cdots = v(r_{i,m}^{\varepsilon}) = \varepsilon.
    \end{equation}
\end{enumerate}
So, there are $2n-1 + m(n-1)$ resources total.
We assume \mbox{$\mathcal{A}_1 = \{r_1\}$}, and for each agent $i \in N \setminus \{1\}$,
\begin{equation}
    \mathcal{A}_i = \{r_1, r_i, r_{i + n-1}, r_{i,1}^{\varepsilon}, \ldots, r_{i,m}^{\varepsilon}\},
\end{equation}
which corresponds to resource values $\leftcurly 1, 1-\alpha,\alpha, \varepsilon, \ldots, \varepsilon \rightcurly$.
We denote this game by $G^\star(m, n)$, which we often write simply as $G^\star$; an example is given in Fig.~\ref{fig:Example_Game}.

Now, $G^\star(m, n)$ admits a stationary distribution under log-linear learning with $\tau(G^\star(m, n)) > 0$, which we often write as $\tau$ for compactness.
For $S \in \{I, B\}$, denote the expected contribution of resource $r$ as $f_\tau^S(r) = v(r) \cdot \pi_\tau^S(r)$, where $\pi_\tau^S(r)$ is the probability~$r$ is selected by an agent in~$G^\star_S$.
The $1$-valued resource is always covered, and resources of value $1-\alpha$, $\alpha$, and $\varepsilon$ are agent-specific, so their expected contributions can be summed across agents.
Thus the expected value of the objective function for $G^\star_S$ is:
\begin{align}
    \nonumber
    E_{a \sim \pi_\tau^{G^\star_S}}[W(a)] \hspace{-0.5mm}=\hspace{-0.5mm} 1 \hspace{-0.5mm}+\hspace{-0.5mm} (\hspace{-0.125mm}n \hspace{-0.5mm}-\hspace{-0.5mm} 1\hspace{-0.125mm}) (\hspace{-0.125mm}f_\tau^S(\hspace{-0.125mm}1 \hspace{-0.5mm}-\hspace{-0.5mm} \alpha\hspace{-0.125mm}) \hspace{-0.5mm}+\hspace{-0.5mm} f_\tau^S(\hspace{-0.125mm}\alpha\hspace{-0.125mm}) \hspace{-0.5mm}+\hspace{-0.5mm} m f_\tau^S(\hspace{-0.125mm}\varepsilon\hspace{-0.125mm})).
\end{align}

\section{Robust Temperature Selection}

First, we study whether a system operator can use a single approach that performs well under nominal communication while avoiding arbitrarily poor performance when communication is compromised.
Theorem~\ref{thm:exists_tradeoff_temperature} shows that for any game $G \in \mathcal{G}$, there exists a temperature under log-linear learning such that nominal efficiency is near-optimal and compromised efficiency is strictly positive.


\begin{theorem}
    \label{thm:exists_tradeoff_temperature}
    Let $G \in \mathcal{G}$ be a multiagent singleton set cover game where agents are assigned the marginal-contribution utility function in~\eqref{eq:marginal_contribution} and update actions according to log-linear learning in~\eqref{eq:LLL}.
    Then, for any $\eta > 0$, there exists temperature ${\tau^\ast} > 0$ and $\delta > 0$ such that the efficiency, given by~\eqref{eq:Efficiency}, for the nominal game $G_I$ and compromised game $G_B$, respectively, satisfies
    \begin{align}
    \Eff_{\tau^\ast} \leftparen G_I \rightparen &\geq 1-\eta,
    \\
    \Eff_{\tau^\ast} \leftparen G_B \rightparen &\geq \delta.
    \end{align}
\end{theorem}

\textit{Proof:}
Let $G \in \mathcal{G}$, and fix $\eta > 0$.
For $a^* \in \arg\max_{a \in \mathcal{A}} W(a)$, \mbox{$W(a^*) > 0$}.
First, we show nominal game $G_I$ admits a temperature that achieves an expected efficiency of at least $1-\eta$.
Since all agents are informed in $G_I$, each agent uses the marginal-contribution utility function.
Thus, for any agent $i \in N$, any $a_{-i} \in \mathcal{A}_{-i}$, and any $a_i,a_i' \in \mathcal{A}_i$, we have
\begin{align}
    \nonumber
    U_i(a_i, a_{-i})-U_i(a_i',a_{-i}) &=
    \leftparen W(a_i, a_{-i})-W(\emptyset,a_{-i}) \rightparen
    \\ \nonumber &-
    \leftparen W(a_i', a_{-i})-W(\emptyset,a_{-i}) \rightparen
    \\ \nonumber &=
    W(a_i, a_{-i})-W(a_i', a_{-i}).
\end{align}
Hence, $G_I$ is an exact potential game with potential function~$W$.
Since $G_I$ has finite agents and resources, it is also finite.
Thus, for every $a \in \mathcal{A}$, known properties of log-linear learning~\citep{blume1993statistical,marden2012revisiting} imply that the stationary distribution is:
\begin{equation}
    \nonumber
    \pi_\tau^{G_I}(a) = \frac{e^{\frac{1}{\tau} W(a)}}{\sum_{a' \in \mathcal{A}} e^{\frac{1}{\tau} W(a')}}.
\end{equation}
As $\tau \to 0$, $\pi_\tau^{G_I}(a)$ concentrates on the action profiles that maximize $W$.
Since $a^* \in \arg\max_{a \in \mathcal{A}} W(a)$, it follows that
\begin{equation}
    \nonumber
    \lim_{\tau\to 0}
    E_{a \sim \pi_\tau^{G_I}}[W(a)] = W(a^*).
\end{equation}
Hence, $\lim_{\tau \to 0} \Eff_\tau\leftparen G_I \rightparen = 1$.
Thus, there exists a temperature ${\tau^\ast} > 0$ sufficiently small such that \mbox{$\Eff_{\tau^\ast}\leftparen G_I \rightparen \geq 1-\eta$}.
Now, consider compromised game $G_B$ under ${\tau^\ast} > 0$.
Once again, log-linear learning on $G_B$ admits a stationary distribution $\pi_{\tau^\ast}^{G_B}$.
Since \mbox{${\tau^\ast}>0$}, every action in each agent's action set is selected with positive probability when that agent updates.
Thus, every action profile has positive probability under $\pi_{\tau^\ast}^{G_B}$.
Hence, $\pi_{\tau^\ast}^{G_B}(a^*) > 0$.
Since \mbox{$W(a^*) > 0$} and \mbox{$W(a) \geq 0$} for all $a \in \mathcal{A}$, we have $E_{a\sim \pi_{\tau^\ast}^{G_B}}[W(a)] \geq W(a^*) \pi_{\tau^\ast}^{G_B}(a^*) > 0$.
Thus,
\begin{equation}
    \nonumber
    \Eff_{\tau^\ast}\leftparen G_B \rightparen = \frac{E_{a\sim \pi_{\tau^\ast}^{G_B}}[W(a)]}{W(a^*)} > 0.
\end{equation}
Now, define $\delta \coloneqq \Eff_{\tau^\ast}\leftparen G_B \rightparen$.
Thus $\delta>0$, so ${\tau^\ast}$ satisfies $\Eff_{\tau^\ast}\leftparen G_I \rightparen \geq 1-\eta$ and $\Eff_{\tau^\ast}\leftparen G_B \rightparen \geq \delta$. \hfill $\blacksquare$

Intuitively, Theorem~\ref{thm:exists_tradeoff_temperature} informs the system operator that for a specific game, randomization can be tuned to protect against compromised communication without sacrificing nominal performance.
In particular, choosing a small but positive temperature allows agents to remain close to optimal behavior under nominal communication, while still assigning positive probability to high-value action profiles when communication is compromised.
Fig.~\ref{fig:Positive_Result} illustrates this tradeoff, showing there exists a temperature that can achieve high nominal efficiency while keeping compromised efficiency bounded away from zero.

\begin{figure}[t]
    \vspace{2mm}
        \begin{center}
             \includegraphics[width=0.45\textwidth]{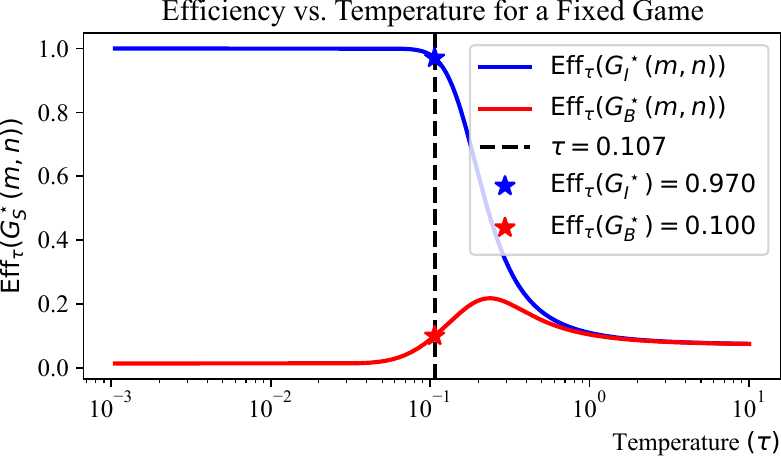}
        \end{center}
        \caption{Efficiency as a function of temperature for $G^\star(m, n)$, where $n = 100$, $m = 25$, $\alpha = 0.25$, and $\varepsilon = 0.01$.
        We plot nominal and compromised efficiency (the upper and lower plot, respectively) for $\tau \in [10^{-3}, 10]$.
        The dashed line at $\tau = 0.107$ produces near-optimal nominal performance, and compromised performance is bounded away from $0$.
        Here, nominal efficiency decreases~$3\%$ from optimal, while compromised efficiency increases~$10\%$ from its worst.
        Thus, even in a game with many negligible-value resources, a small amount of randomization obtains a robustness benefit at little cost to nominal performance.}
        \label{fig:Positive_Result}
    \vspace{-1mm}
\end{figure}

\section{Compromised Efficiency Breaks Down with Large Action Sets}

Theorem~\ref{thm:exists_tradeoff_temperature} shows that, for any game, a system operator can select a fixed temperature so that nominal efficiency is near-optimal and compromised efficiency is strictly positive.
However, the strength of the positive compromised-efficiency guarantee depends on the relationship between agents and resources.
Theorem~\ref{thm:temperature_selection_arbitrarily_bad} shows that this dependency is unavoidable: if agents are given too many resources, then any temperature selected admits a game whose compromised efficiency is arbitrarily poor.

\begin{theorem}
    \label{thm:temperature_selection_arbitrarily_bad}
    Let $\tau:\mathcal{G}\to\mathbb{R}_{>0}$ be a temperature selection scheme for multiagent singleton set cover games where agents are assigned the marginal-contribution utility function in~\eqref{eq:marginal_contribution} and update according to log-linear learning in~\eqref{eq:LLL}.
    For any $\delta > 0$, there exists $G \in \mathcal{G}$ such that
    \begin{equation}
        \label{eq:temperature_selection_arbitrarily_bad}
        \Eff_{\tau(G)} \leftparen G_B \rightparen < \delta.
    \end{equation} 
\end{theorem}

\textit{Proof:}
We show that for any~$\tau$, we can construct a game $G^\star$ as defined in Section~\ref{sec:star_game} such that compromised efficiency is arbitrarily poor.
Fix $\tau:\mathcal{G}\to\mathbb{R}_{>0}$ and $\delta>0$.
Let $\alpha \in (0,1/2]$ and choose $\varepsilon>0$ such that
\begin{equation}
    \label{eq:choose_epsilon_small}
    \frac{\varepsilon}{1-\alpha} < \frac{\delta}{4}.
\end{equation}
For each $m,n$, define $\tau_{m,n} \coloneqq \tau\leftparen G^\star(m, n) \rightparen$.
The proof proceeds by showing that for sufficiently large $n$, compromised game $G_B^\star(m,n)$ has efficiency strictly less than~$\delta$ for some~$m$.
Now, choose a fixed $n_0$ large enough such that
\begin{equation}
    \label{eq:choose_n_temperature_selection}
    \frac{1 + (n_0-1)\varepsilon}{1 + (n_0-1)(1-\alpha)} < \frac{\delta}{2},
\end{equation}
and consider the sequence of temperatures $\leftcurly \tau_{m, n_0} \rightcurly_{m=1}^ \infty$.
From here, the proof is completed in two parts.

\textit{Part~1:} Suppose there exists $c>0$ and an infinite sequence of~$m$ values such that \mbox{$\tau_{m, n_0} \geq c$}.
Now, the probabilities that a blind agent selects the $(1-\alpha)$-valued resource and the $\alpha$-valued resource are
\begin{align}
    \label{eq:pi_B_1_minus_alpha}
    \pi_{\tau_{m, n_0}}^B\hspace{-1mm}(1 \hspace{-0.25mm}-\hspace{-0.25mm} \alpha) \hspace{-0.25mm}&=\hspace{-0.25mm}
    \frac{e^{\frac{1 \hspace{-0.25mm}-\hspace{-0.25mm} \alpha}{\tau_{m, n_0}}}}{e^{\frac{1}{\tau_{m, n_0}}} \hspace{-0.5mm}+\hspace{-0.5mm} e^{\frac{1 \hspace{-0.25mm}-\hspace{-0.25mm} \alpha}{\tau_{m, n_0}}} \hspace{-0.5mm}+\hspace{-0.5mm} e^{\frac{\alpha}{\tau_{m, n_0}}} \hspace{-0.5mm}+\hspace{-0.5mm} m e^{\frac{\varepsilon}{\tau_{m, n_0}}}},
    \\ \label{eq:pi_B_alpha}
    \pi_{\tau_{m, n_0}}^B(\alpha) &=
    \frac{e^{\frac{ \alpha}{\tau_{m, n_0}}}}{e^{\frac{1}{\tau_{m, n_0}}} \hspace{-0.5mm}+\hspace{-0.5mm} e^{\frac{1 \hspace{-0.25mm}-\hspace{-0.25mm} \alpha}{\tau_{m, n_0}}} \hspace{-0.5mm}+\hspace{-0.5mm} e^{\frac{\alpha}{\tau_{m, n_0}}} \hspace{-0.5mm}+\hspace{-0.5mm} m e^{\frac{\varepsilon}{\tau_{m, n_0}}}}.
\end{align}
Since $\tau_{m, n_0}\geq c$, $e^{\frac{1}{\tau_{m, n_0}}}$, $e^{\frac{1-\alpha}{\tau_{m, n_0}}}$, and $e^{\frac{\alpha}{\tau_{m, n_0}}}$ are bounded above and $m e^{\frac{\varepsilon}{\tau_{m, n_0}}}$ is unbounded above as \mbox{$m \to  \infty$}.
Thus, $\pi_{\tau_{m, n_0}}^B(1-\alpha)\to 0$ and $\pi_{\tau_{m, n_0}}^B(\alpha)\to 0$.
Moreover,
\begin{equation}
    \nonumber
    m\pi_{\tau_{m, n_0}}^B(\varepsilon) = \frac{m e^{\frac{\varepsilon}{\tau_{m, n_0}}}}{e^{\frac{1}{\tau_{m, n_0}}} + e^{\frac{1-\alpha}{\tau_{m, n_0}}} + e^{\frac{\alpha}{\tau_{m, n_0}}} + m e^{\frac{\varepsilon}{\tau_{m, n_0}}}} \to 1,
\end{equation}
as $m \to  \infty$ along this sequence.
Thus,
\begin{equation}
    \nonumber
    \lim_{m\to \infty} E_{a\sim\pi_{\tau_{m, n_0}}^{G_B^\star(m, n_0)}}[W(a)] = 1 + (n_0-1)\varepsilon.
\end{equation}
Since the optimal objective value is
\begin{equation}
    \nonumber
    W(a^*) = 1 + (n_0-1)(1-\alpha),
\end{equation}
it follows from~\eqref{eq:choose_n_temperature_selection} that for sufficiently large~$m$,
\begin{equation}
    \nonumber 
    \Eff_{\tau_{m, n_0}}\leftparen G_B^\star(m, n_0) \rightparen < \delta.
\end{equation}

\textit{Part~2:} Suppose no such $c>0$ exists.
Then $\tau_{m, n_0}\to 0$ as $m\to \infty$.
By~\eqref{eq:pi_B_1_minus_alpha}, the probability a blind agent selects the $(1-\alpha)$-valued resource satisfies $\pi_{\tau_{m, n_0}}^B(1-\alpha) \leq e^{-\frac{\alpha}{\tau_{m, n_0}}}$.
Since $\tau_{m, n_0}\to 0$, it follows that $\pi_{\tau_{m, n_0}}^B(1-\alpha)\to 0$.
Similarly, by~\eqref{eq:pi_B_alpha}, $\pi_{\tau_{m, n_0}}^B(\alpha) \leq e^{-\frac{1-\alpha}{\tau_{m, n_0}}}$, so $\pi_{\tau_{m, n_0}}^B(\alpha)\to 0$ as $\tau_{m, n_0}\to 0$.
Finally, the total probability that a blind agent selects one of its $\varepsilon$-valued resources is at most $1$, so
\begin{equation}
    \nonumber
    m f_{\tau_{m, n_0}}^B(\varepsilon) = m \leftparen \varepsilon \pi_{\tau_{m, n_0}}^B(\varepsilon) \rightparen \leq \varepsilon.
\end{equation}
Thus, $\lim_{m\to \infty} E_{a\sim\pi_{\tau_{m, n_0}}^{G_B^\star(m, n_0)}}[W(a)] \leq 1 + (n_0-1)\varepsilon$.
Using~\eqref{eq:choose_n_temperature_selection} again, there exists sufficiently large~$m$ such that
\begin{equation}
    \nonumber 
    \Eff_{\tau_{m, n_0}}\leftparen G_B^\star(m, n_0) \rightparen < \delta.
\end{equation}
Hence, for any $\tau$ and any \mbox{$\delta>0$}, there exists a game $G^\star(m, n) \in \mathcal{G}$ such that \mbox{$\Eff_{\tau(G^\star(m, n))}\leftparen G_B^\star(m,n) \rightparen < \delta$}. \hfill $\blacksquare$

Intuitively, while Theorem~\ref{thm:exists_tradeoff_temperature} shows that for any fixed game, a system operator can choose a temperature that keeps compromised efficiency bounded away from zero, Theorem~\ref{thm:temperature_selection_arbitrarily_bad} shows that no universal lower bound exists that can be applied to any game.
Further, the proof of Theorem~\ref{thm:temperature_selection_arbitrarily_bad} demonstrates that  this limitation arises if agents are given too many resources to consider, then the compromised-efficiency guarantee can become arbitrarily small.
The following section illustrates \textit{Part~1} and \textit{Part~2} of the proof, and provides numerical experiments to demonstrate how compromised efficiency degrades on the star-graph game used in the proof.

\subsection{Numerical Experiments}\label{subsec:Num_Experiments}

\begin{figure}[t]
    \vspace{2mm}
        \begin{center}
             \includegraphics[width=0.45\textwidth]{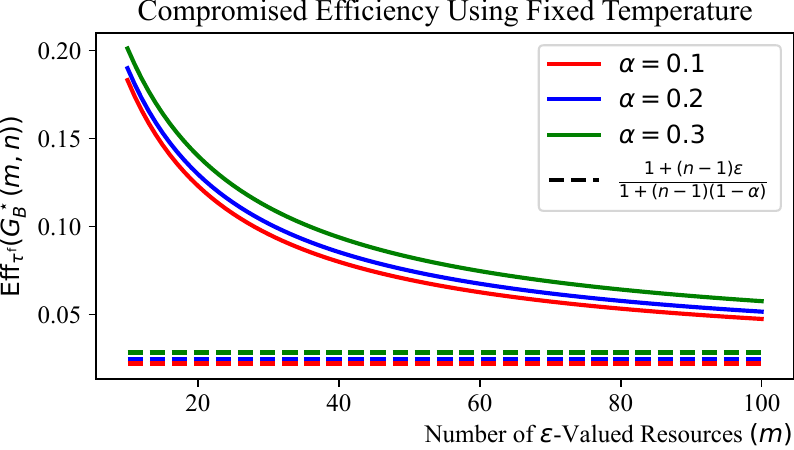}
        \end{center}
        \caption{Compromised efficiency using fixed temperature as a function of $\varepsilon$-valued resources.
        For $G^\star(m, n)$, we set $n = 100$ and $\varepsilon = 0.01$.
        For each $\alpha \in \{0.1, 0.2, 0.3\}$, we plot $\mathrm{Eff}_{\tau^{\mathrm{f}}}(G_B^\star(m, n))$ for $m \in \{10, \ldots, 100\}$, where $\tau^{\mathrm{f}}(G^\star(m, n)) = 1-\alpha$.
        The dashed lines show the asymptotic lower bound, $\frac{1 + (n-1) \varepsilon}{1 + (n-1)(1-\alpha)}$, of efficiency as the number of $\varepsilon$-valued resources grows $(m)$.
        The plot illustrates \textit{Part~1} of the proof of Theorem~\ref{thm:temperature_selection_arbitrarily_bad}, which shows that if temperature has an asymptotic lower bound (i.e., $\tau \geq c > 0$), then blind agents select a $\varepsilon$-valued resource with increasingly high probability, causing efficiency to degrade as $m$ increases.}
        \label{fig:Negative_Result_Part_1}
    \vspace{1mm}
\end{figure}

Here, we consider two temperature selection schemes, then show that games exist where agents have arbitrarily poor compromised efficiency when they randomize according to the selected temperature.
In both cases, poor performance arises because agents are asked to choose from too many resources, increasing the likelihood they choose resources that provide negligible system-level benefit.

For a game $G^\star(m, n)$, agents want to select resources with values $1-\alpha$ and $\alpha$ with high probability, and select the $\varepsilon$-valued resources with much lower probability.
To model this, we consider a ``fixed'' temperature selection scheme that scales with the payoffs of the actions agents want to prioritize.
We also consider a ``cardinality-based'' temperature selection scheme that builds on the fixed scheme, and accounts for low-value resources by decreasing logarithmically as a function of the number of available resources.
Thus, the two schemes are given by
\begin{align}
    \label{eq:fixed_temperature_selection_numerics}
    \tau^{\mathrm{f}} \leftparen G^\star(m, n) \rightparen &\coloneqq 1-\alpha,
    \\
    \label{eq:cardinality_temperature_selection_numerics}
    \tau^{\mathrm{c}} \leftparen G^\star(m, n) \rightparen &\coloneqq \frac{1-\alpha}{\ln(m + 3)}.
\end{align}

The fixed scheme ensures that resources with value $1-\alpha$ or $\alpha$ are selected with exponentially higher probability than individual $\varepsilon$-valued resources.
However, an issue arises when the number of $\varepsilon$-valued resources becomes large.
Even though each $\varepsilon$-valued resource is selected with low probability, the total probability that an agent selects some $\varepsilon$-valued resource increases as a function of~$m$.
Consequently, under the scheme in~\eqref{eq:fixed_temperature_selection_numerics}, blind agents place most of their probability mass on negligible resources, causing compromised efficiency to become arbitrarily poor.
This result is presented graphically in Fig.~\ref{fig:Negative_Result_Part_1}.

The cardinality-based scheme adjusts for this issue by decreasing temperature as agents' actions increase.
The temperature decreases logarithmically as~$m$ increases linearly, so agents should not become overly deterministic.
However, a new issue arises where the temperature becomes too small when agents are blind, so they select the $1$-valued resource with increasingly high probability since they do not know it is covered.
Thus, under the scheme in~\eqref{eq:cardinality_temperature_selection_numerics}, agents place most of their probability mass on a covered resource, causing efficiency to degrade under compromised communication.
This is shown in Fig.~\ref{fig:Negative_Result_Part_2}.

\begin{figure}[t]
    \vspace{2mm}
        \begin{center}
             \includegraphics[width=0.45\textwidth]{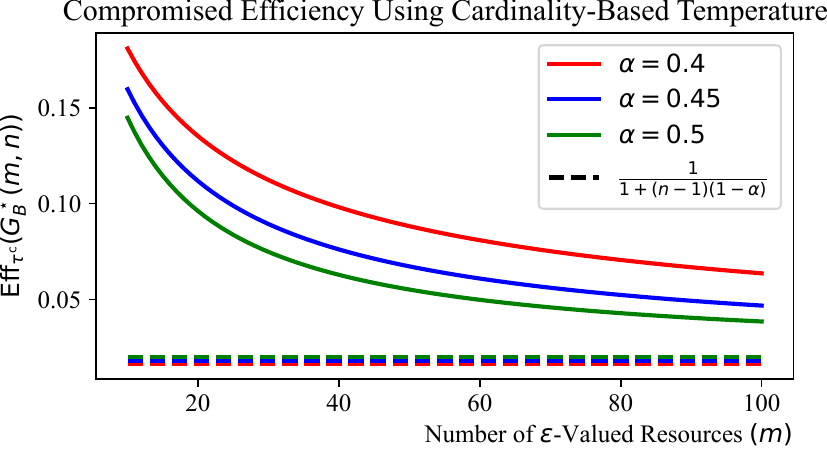}
        \end{center}
        \caption{Compromised efficiency using cardinality-based temperature as a function of $\varepsilon$-valued resources.
        For $G^\star(m, n)$, we set $n = 100$ and $\varepsilon = 0.01$.
        For each \mbox{$\alpha \in \{0.4, 0.45, 0.5\}$}, we plot $\mathrm{Eff}_{\tau^{\mathrm{c}}}(G_B^\star(m, n))$ for \mbox{$m \in \{10, \ldots, 100\}$}, where $\tau^{\mathrm{c}}(G^\star(m, n)) = \frac{(1-\alpha)}{\ln(m + 3)}$.
        The dashed lines show the asymptotic lower bound on efficiency, $\frac{1}{1 + (n-1)(1-\alpha)}$, for the cardinality-based scheme as the number of $\varepsilon$-valued resources grows~$(m)$.
        The plot illustrates \textit{Part~2} of the proof of Theorem~\ref{thm:temperature_selection_arbitrarily_bad}, which shows that if the temperature approaches $0$, then blind agents place increasingly high probability on the $1$-valued resource as $m$ increases.}
        \label{fig:Negative_Result_Part_2}
    \vspace{1mm}
\end{figure}

\section{Conclusion}

We study whether randomization under log-linear learning can balance nominal performance with robustness to compromised communication in multiagent submodular maximization games.
We show that for every game, there exists a temperature such that nominal efficiency is near-optimal while compromised efficiency remains strictly positive.
However, for every temperature selection scheme, there exist games with arbitrarily poor compromised efficiency.
Future work should develop game-dependent randomization rules, identify conditions permitting stronger performance guarantees, and explore whether efficiency improves monotonically as agents regain information under partial communication failures.


\bibliography{ifacconf}             
                                                   

\end{document}